\documentclass[conference]{IEEEtran}
\IEEEoverridecommandlockouts
\usepackage{cite}
\usepackage{amsmath,amssymb,amsfonts,amsthm}
\usepackage{bm}
\usepackage{dsfont}
\usepackage{algorithm2e}
\usepackage{graphicx}
\usepackage{textcomp}
\usepackage{xcolor}
\usepackage{hyperref}
\def\BibTeX{{\rm B\kern-.05em{\sc i\kern-.025em b}\kern-.08em
    T\kern-.1667em\lower.7ex\hbox{E}\kern-.125emX}}

\newcommand{\bv}[1]{\mathbf{#1}}
\newcommand{\E}{\mathbb{E}}
\newcommand{\sX}{\mathcal{X}}
\newcommand{\sU}{\mathcal{U}}

\newtheorem{theorem}{Theorem}[section]
\newtheorem{problem}[theorem]{Problem}

\newtheorem{proposition}[theorem]{Proposition}
\newtheorem{remark}[theorem]{Remark}
    
\begin{document}

%\title{Discrete fully probabilistic design: a tool to design control policies from examples}

\title{Discrete fully probabilistic design: towards a control pipeline for the synthesis of policies from examples}

\author{Enrico Ferrentino, Pasquale Chiacchio, Giovanni Russo
\thanks{All authors are with the Department of Computer and Electrical Engineering and Applied Mathematics (DIEM), University of Salerno, 84084 Fisciano, SA, Italy, e-mail: \{eferrentino,pchiacchio,giovarusso\}@unisa.it.}
}

%\author{\IEEEauthorblockN{Enrico Ferrentino}
%\IEEEauthorblockA{\textit{Department of Computer Engineering, Electrical Engineering and Applied Mathematics} \\
%\textit{University of Salerno}\\
%Fisciano, Italy \\
%eferrentino@unisa.it}
%\and
%\IEEEauthorblockN{Pasquale Chiacchio}
%\IEEEauthorblockA{\textit{Department of Computer Engineering, Electrical Engineering and Applied Mathematics} \\
%\textit{University of Salerno}\\
%Fisciano, Italy \\
%pchiacchio@unisa.it}
%\and
%\IEEEauthorblockN{Giovanni Russo}
%\IEEEauthorblockA{\textit{Department of Computer Engineering, Electrical Engineering and Applied Mathematics} \\
%\textit{University of Salerno}\\
%Fisciano, Italy \\
%giovarusso@unisa.it}
%}

\maketitle

\begin{abstract}
We present the principled design of a control pipeline for the synthesis of policies from examples data. The pipeline, based on a discretized design which we term as discrete fully probabilistic design, expounds an algorithm recently introduced in \cite{Gagliardi} to synthesize policies from examples for constrained, stochastic and nonlinear systems. Contrary to other approaches, the pipeline we present: (i) does not need the constraints to be  fulfilled in the possibly noisy example data; (ii) enables control synthesis even when the data are collected from an example system that is different from the one under control. The design is benchmarked numerically on an example that involves controlling an inverted pendulum with actuation constraints starting from data collected from a physically different pendulum that does not satisfy the system-specific actuation constraints. We also make our fully documented code openly available. 
%We present a discretized design that expounds an algorithm recently introduced in \cite{Gagliardi} to synthesize control policies from examples for constrained, possibly stochastic and nonlinear, systems. The constraints do not need to be  fulfilled in the possibly noisy example data, which in turn might be collected from a system that is different from the one under control. For this discretized design,  we discuss a number of  properties and give a design pipeline. The design, which we term as discrete fully probabilistic design, is benchmarked numerically on an example that involves controlling an inverted pendulum with actuation constraints starting from data collected from a physically different pendulum that does not satisfy the system-specific actuation constraints.
\end{abstract}


\section{Introduction}\label{sec:introduction}
Over the past few years, much research effort has been devoted to the problem of synthesizing control polices directly from data, bypassing the need to devise and identify a mathematical model in the form of difference/differential equations \cite{HOU20133}. An appealing {\em data-driven} control framework is that of designing controllers by using example data \cite{Han_lIU_zHU_Pas_19,Gagliardi,makdah2021learning,tu2021sample}. Within this {\em control from examples framework}, one seeks to synthesize policies so that the closed-loop system {\em tracks} some desired behavior extracted from the examples. In this context, a key challenge is that of developing an end-to-end pipeline to synthesize control policies for nonlinear, stochastic and constrained systems from noisy example data. Towards this aim, we present the principled design of a pipeline that enables control synthesis in these situations. With our pipeline, situations can be considered when data are collected from a system that is different from the one under control and does not satisfy the system-specific constraints, which might hence be unknown when the examples are collected. We refer to \cite{10.1145/3054912,doi:10.1146/annurev-control-100819-063206,cite-key,GARRABE202281} for detailed interdisciplinary surveys on the related topics of (imitation) learning and sequential decision making and we now briefly survey a number of works that are related to the specific methodological framework we leverage in this paper.

\subsubsection*{Related work} 
We leverage a Bayesian approach \cite{Peterka_V_Bayesian_Approach_to_sys_ident_1981} to dynamical systems that allows to represent the behaviors of these systems via probability functions. In the context of control, the approach has been leveraged in e.g.\ \cite{Karny96,Karny_M+Guy_TV_Sys&Ctr_lett_2006,KARNY2012105,Herzallah_R_JNeurNet_2015,10.1007/978-3-030-01713-2_20,Pegueroles_G+Russo_G_ECC19_confid} for the design of randomized control policies that enable tracking of a given target behavior. In these papers, the tracking problem is tackled by setting-up an unconstrained optimization problem. Closely related works, include {\cite{Todorov11478,NIPS2006_d806ca13,Kappen_2012}}. These works, by leveraging a similar framework, formalize the control problem as the (unconstrained) problem of minimizing a cost that captures the discrepancy between an ideal probability density function and the actual probability density function of the system under control.  An online version of these algorithms has been proposed in \cite{6716965}: in such a work, by leveraging an average cost formulation, the probability mass function for the state transitions is found.  {Finally, we also recall here  \cite{9244209,9446558}, where policies are obtained from the minimization of similar costs by leveraging multiple, specialized, datasets. These last papers, together with \cite{Gagliardi}, introduce constraints to the probabilistic formulation. Finally, we also recall \cite{9349120}, where a decision making architecture for Robust Model-Predictive Path Integral Control is proposed and this allows the introduction of the constraints (see also references therein).}

\subsubsection*{Statement of contributions and organization of the paper}
We introduce, and benchmark, the principled design of a pipeline for the control synthesis from examples data. The pipeline expounds an algorithm presented in \cite{Gagliardi} for the synthesis of control policies from examples. We term this discretized design {\em discrete fully probabilistic design} (DFPD). While \cite{Gagliardi} seeks to find an analytical solution to the control problem, here we rely on finding a purely numerical solution, exploiting convexity of the underlying optimization problems that, as we show, need to be solved in order to synthesize the policy. We also discuss a number of properties of DFPD. In contrast to other works on imitation learning and control synthesis from example, by exploiting \cite{Gagliardi}, our design: (i) does not need the constraints to be  fulfilled in the possibly noisy example data; (ii) enables control synthesis even when the data are collected from an example system that is different from the one under control. Finally, the design is numerically benchmarked on an inverted pendulum and we also make the code openly available and fully documented (see {\url{https://github.com/unisa-acg/discrete-fpd}}).

\noindent The paper is organized as follows. After giving the mathematical preliminaries and formalizing the statement of the problem (Section \ref{sec:problem_statement}), we present the DFPD. This is done in Section \ref{sec:algorithm}, where we also discuss a number of its properties. In Section \ref{sec:pipeline} we describe a design pipeline to use DFPD. The pipeline illustrates a process to determine, via DFPD, control inputs from the data. Finally, the effectiveness of DFPD is illustrated in Section \ref{sec:numerical_results} on a pendulum with actuation constraints. Concluding remarks are given  in Section \ref{sec:conclusions}.

\section{Mathematical Set-up and Problem Statement}
\label{sec:problem_statement}

We denote sets via {\em calligraphic} capital characters and vectors  in 
{\bf bold}. Multidimensional random
variables and their realizations are both denoted by lower-case bold letters. All the random variables we consider are discrete and sampled from probability mass functions (pmfs) with compact supports. Throughout the paper, we also refer to pmfs as (normalized) {\em histograms} and to its domain as (discrete) {\em alphabet}. The notation $\bv{z} \sim P(\bv{z})$ denotes the fact that the random variable $\bv{z}$ is sampled from the pmf $P(\bv{z})$. Given the set $\mathcal{Z} \subseteq \mathbb{R}^{n_z}$, its \emph{indicator function} is denoted by $\mathds{1}_{\mathcal{Z}}({\bf z})$, so that $\mathds{1}_{\mathcal{Z}}({\bf z}) = 1\, \forall {\bf z} \in \mathcal{Z}$ and $0$ otherwise. We also let $\E_P[\bv{h}(\bv{z})]:=\sum P(\bv{z})\bv{h}(\bv{z})$, where the sum is implicitly assumed to be taken over the support of $P(\bv{z})$,  be the expectation of a function, say $\bv{h}(\cdot)$, of $\bv{z}$.  The joint pmf of $\bv{z}_1$ and $\bv{z}_2$ is denoted by $P(\bv{z}_1,\bv{z}_2)$ and the conditional pmf of $\bv{z}_1$ given $\bv{z}_2$ is denoted by $P(\bv{z}_1|\bv{z}_2)$. The control problem considered in this paper is stated in terms of the Kullback-Leibler (or simply KL) divergence \cite{Kullback}. Given the histograms $P(\boldsymbol{\alpha})$ and $Q(\boldsymbol{\alpha})$ defined over the discrete alphabet $\mathcal{A}$, the KL-divergence is defined as
\begin{equation} \label{dkl_definition}
\mathcal{D}_{KL}(P || Q) := \sum_{\boldsymbol{\alpha} \in \mathcal{A}} P(\boldsymbol{\alpha}) \ln \frac{P(\boldsymbol{\alpha})}{Q(\boldsymbol{\alpha})}.
\end{equation}
%\GR{Introdurre funzione indicatore come in paper Gagliardi}
%\EF{Fatto}

\subsection{Formulation of the control problem}

Let ${\bf x} \in\sX\subset \mathbb{R}^{d_x}$ be the state of a dynamical system, and ${\bf u} \in\sU\subset \mathbb{R}^{d_u}$ be the control variable. The time variable is $t \in [0,T]$, $T<+\infty$. The time interval $[0,T]$ is discretized in $n+1$ instants with step $\Delta t >0$. This yields the sequence of time-instants $t(k) = k \Delta t$, with $k = 0, \ldots, n$ and $T = n \Delta t$. We also let ${\bf x}(k)$ and ${\bf u}(k)$ be ${\bf x}$ and ${\bf u}$ evaluated in the discretized time. We recall that $\sX$ and $\sU$ are both compact and we set
\begin{equation} \label{sets}
\mathcal{X} := \lbrace {\bf x}_0, \ldots, {\bf x}_{m-1} \rbrace, \ \ \ \mathcal{U} := \lbrace {\bf u}_0, \ldots, {\bf u}_{z-1} \rbrace ~\forall k.
\end{equation}
For notational convenience, we use the shorthand notation $\bv{x}_j(k)$ (resp.\ $\bv{u}_h(k)$) to denote that $\bv{x}$ (resp.\ $\bv{u}$) at time-instant $k$ equals the value of the $j$-th (resp.\ $h$-th) element in $\sX$ (resp.\ $\sU$).

Following \cite{Peterka_V_Bayesian_Approach_to_sys_ident_1981,Karny96} one can describe the behavior of a given system by computing the joint pmf $P^{n} := P\big({\bf x}(0), \ldots, {\bf x}(n), {\bf u}(0), \ldots, {\bf u}(n-1)\big)$, obtained from e.g.\ state-input data collected from the system. Likewise, one can also define a desired (or reference) behavior for the system in terms of a joint pmf, say $Q^{n}:=Q\big({\bf x}(0), \ldots, {\bf x}(n), {\bf u}(0), \ldots, {\bf u}(n-1)\big)$. Then, by making the standard Markov's assumption, the chain rule for pmfs yields the following convenience factorization:
\begin{equation} \label{conditional_probabilities}
\begin{split}
P^{n} = \prod_{k=1}^{n} & P \big( {\bf x}(k) | {\bf x}(k-1), {\bf u}(k-1)\big) \\ 
& P \big( {\bf u}(k-1) | {\bf x}(k-1) \big) P(\bv{x}(0)) \\
=: \prod_{k=1}^{n} & \tilde{P}_X^k \tilde{P}_U^{k-1} P(\bv{x}(0)) =: \prod_{k=1}^{n} \tilde{P}^k P(\bv{x}(0)),\\
Q^{n} = \prod_{k=1}^{n} & Q \big( {\bf x}(k) | {\bf x}(k-1), {\bf u}(k-1) \big) \\ 
& Q \big( {\bf u}(k-1) | {\bf x}(k-1) \big)Q(\bv{x}(0)) \\
=: \prod_{k=1}^n & \tilde{Q}_X^k \tilde{Q}_U^{k-1}Q(\bv{x}(0)) =:  \prod_{k=1}^{n} \tilde{Q}^kQ(\bv{x}(0)).
\end{split}
\end{equation}
%\GR{Chiamarente nella (3) pima non avevi le joint che pure avevi in precedenza definito giusto sopra l'equazione. Ho modificato l'equazione per correggere l'errore. Da controllare se era questo quello che intendevi. Se non e' cosi' ci sono cambiamenti pesanti da fare alla notazione.}
The former term in the first product of the definitions takes the name of \emph{state evolution model}, while the latter is the \emph{randomized control law} \cite{Karny06}. In the context of this paper, the pmf $Q^{n}$ is collected from example data. This allows us to state the problem of synthesizing control policies from example data for systems with actuation constraints as follows.

\begin{problem}[Global optimization problem] \label{prb:global_opt}
Find a solution to the optimization problem
\begin{equation} \label{generic_opt_problem}
\begin{split}
 \min_{\tilde{P}_U^{0}, \ldots, \tilde{P}_U^{n-1}} & \mathcal{D}_{KL}(P^n || Q^n) \\
\mbox{s.t.~} & {\bf f}_k(p^k_{hi}) \leq {\bf 0} ~\forall k \\
& {\bf g}_k(p^k_{hi}) = {\bf 0} ~\forall k
\end{split}
\end{equation}
where ${\bf f}_k(p^k_{hi})$, ${\bf g}_k(p^k_{hi})$ are constraint vector functions defining a convex feasibility domain and $p^k_{hi} := {P}\left({\bf u}_h(k)|{\bf x}_i(k)\right)$.
\end{problem}
For Problem \ref{prb:global_opt} we make the following observations. Minimizing the KL-divergence amounts at minimizing the discrepancy between $P^n$ and $Q^n$. These pmfs can be obtained by observing different systems and hence the formulation embeds the possibility of synthesizing policies for a given system by using examples collected on a different one. 

\section{The Discrete Fully Probabilistic Design Algorithm}\label{sec:algorithm}

The pseudo-code for the DFPD is given in Algorithm \ref{alg:discrete_fpd}. The DFPD takes as input the probabilistic descriptions $\tilde{P}_X^{k+1}\big({\bf x}(k+1)|{\bf x}_i, {\bf u}_h\big)$,  $\tilde{Q}_X^{k+1}\big({\bf x}(k+1)|{\bf x}_i, {\bf u}_h\big)$,  $\tilde{Q}_U^{k}\big({\bf u}(k)|{\bf x}_i\big)$ and the constraints of Problem \ref{prb:global_opt} (optional). DFPD then outputs the optimal solution to Problem \ref{prb:global_opt}.\footnote{As noted in \cite{Gagliardi}, given the convexity of the cost, if the constraints define a convex set, the existence of the solution is guaranteed.} The core of the procedure is a backward recursion that solves, at each iterate, the following

\begin{problem}[Local optimization problem] \label{prb:local_opt}
Find a solution, for each ${\bf x}_i \in \mathcal{X}$, to the following
\begin{equation} \label{opt_problem_constrained}
\begin{split}
\min_{p_{1i},\ldots,p_{zi}} & d\left({\bf x}_i\right) = p_{1i} \ln p_{1i} + p_{1i} (d_{1i}^x + r_{1i} - a_{1i}) + \ldots + \\
& + p_{zi} \ln p_{zi} + p_{zi} (d_{zi}^x + r_{zi} - a_{zi}) \\
\mbox{s.t.} & \sum_{h=1}^z p_{hi} = 1 \\
 & p_{hi} \geq 0 ~ \forall h \\
 & {\bf f}_k(p_{hi}) \leq {\bf 0} \\
 & {\bf g}_k(p_{hi}) = {\bf 0}
\end{split}
\end{equation}
where $p_{hi} = p_{hi}^k$ are the probabilities at the current iterate, $d^x_{hi} = \mathcal{D}_{KL}(\tilde{P}_X^{k+1} || \tilde{Q}_X^{k+1})$, $a_{hi} = \ln\left({Q}\left({\bf u}_h(k)|{\bf x}_i(k)\right)\right)$ and
\begin{equation} \label{recursive_constant}
r_{hi} = \sum_{j=0}^{m-1} P\big( {\bf x}_j (k+1) | {\bf x}_i(k), {\bf u}_h(k) \big) d \big( {\bf x}_j(k+1) \big).
\end{equation}
\end{problem}

At each step of the optimization horizon, and for each state ${\bf x}_i \in \mathcal{X}$, Algorithm \ref{alg:discrete_fpd} computes the scalars $d_{hi}$, $a_{hi}$ and $r_{hi}$, needed for the resolution of Problem \ref{prb:local_opt}. We note that, for the $k$-th step and the $i$-th state, the computation of $r_{hi}$ requires the recursion of the optimal cost $d({\bf x}(k+1))$ \emph{for all states}. Therefore, at $k$, each single cost $d({\bf x}_i(k))$ is temporarily stored in the auxiliary variable $c({\bf x}_i)$ in order to be reused at $k-1$. We now give the following
\begin{proposition}
Algorithm \ref{alg:discrete_fpd} returns the optimal solution to Problem \ref{prb:global_opt}.
\end{proposition}
\begin{proof}
The proof follows the technical derivations of \cite{Gagliardi} where an analytical solution to Problem \ref{alg:discrete_fpd} is given via a backward recursion and a primal-dual argument. The  difference is that Algorithm \ref{alg:discrete_fpd} numerically solves, at each $k$, a reformulation of the problem solved in the paper mentioned above. This reformulation is in fact Problem \ref{prb:local_opt} and, for completeness, a sketch of its derivation is given in Appendix \ref{apx:derivation}.
\end{proof}

\begin{algorithm}
\caption{Discrete fully probabilistic design algorithm}
\label{alg:discrete_fpd}
{\bf Inputs:} 
$\tilde{P}_X^{k+1}\big({\bf x}(k+1)|{\bf x}_i, {\bf u}_h\big)$,  $\tilde{Q}_X^{k+1}\big({\bf x}(k+1)|{\bf x}_i, {\bf u}_h\big)$, $\tilde{Q}_U^{k}\big({\bf u}(k)|{\bf x}_i\big)$, constraints of Problem \ref{prb:global_opt} (optional)\\
{\bf Output:} solution to Problem \ref{prb:global_opt}\\
{\bf Initialize:} $d({\bf x}_i) \gets 0~\forall i = 0, \ldots, m-1$ \\
{\bf Main loop:} \\
\For{$k \gets n-1$ to $0$}
{
\For{\textbf{each} ${\bf x}_i \in \mathcal{X}$}
{
$d^x_{hi} \gets \mathcal{D}_{KL}(\tilde{P}^{k+1}_X || \tilde{Q}^{k+1}_X)~\forall h$ \\
$a_{hi} \gets \ln \left( \tilde{Q}_U^{k}({\bf u}_h|{\bf x}_i) \right) \forall h$ \\
$r_{hi} \gets \sum_{x_j \in \mathcal{X}} \tilde{P}_X^{k+1}({\bf x}_j | {\bf x}_i, {\bf u}_h) d({\bf x}_j)~\forall h$ \\
Find $\tilde{P}_U^k$ by solving Problem \ref{prb:local_opt} with the coefficients above \\
Store optimal cost $c({\bf x}_i)$
}
$d({\bf x}_i) \gets c({\bf x}_i)~\forall i=0,\ldots,m-1$
}
\end{algorithm}

We also make the following
\begin{remark}\label{prp:time_invariance}
Consider a receding horizon set-up with time-invariant constraints and where, at each $k$, the optimization problem solved in the receding horizon window remains the same. Further, note that Algorithm \ref{alg:discrete_fpd}, by construction, finds all the possible (for each state) optimal conditional probability functions $\tilde{P}_U^0$. This implies that, in principle, for this special case, one can find the policy off-line. Specifically, once the problem is solved off-line, the control input can be simply obtained by sampling, at each $k$, from $\tilde{P}_U^0$.
\end{remark}
The rationale behind Remark \ref{prp:time_invariance} is as follows.
Let us assume that $t_0$ is the current time and Problem \ref{prb:global_opt} is solved via Algorithm \ref{alg:discrete_fpd}. In this case, $t=t_0 + k \Delta t$, so that a finite horizon $T$ of $n$ optimization steps is considered. Although the current state ${\bf x}(t_0) = {\bf x}(k = 0)$ is known, Algorithm \ref{alg:discrete_fpd} solves Problem \ref{prb:global_opt} $\forall {\bf x}(t_0) \in \mathcal{X}$. Let us name this optimal solution $P({\bf u}(k)|{\bf x}(k))$. In a receding horizon setup, $P({\bf u}(0)|{\bf x}(0))$ is the control action at $t=t_0$. Now, let us consider a generic time $t_1 > t_0$ at which Problem \ref{prb:global_opt} is solved again with the same strategy, with $t=t_1 + k \Delta t$ and $k=0,\ldots,n-1$ corresponding to the same finite horizon $T$ and with ${\bf x}(t_1) \in \mathcal{X}$. Let us name this optimal solution $P_{t_1}({\bf u}(k)|{\bf x}(k))$. If the probabilistic models and the constraints are time-invariant, then $P_{t_1}({\bf u}(k)|{\bf x}(k)) = P({\bf u}(k)|{\bf x}(k)) \, \forall k$. As a consequence, in our receding horizon setup, the control action at $t=t_1$ is, again, $P({\bf u}(0)|{\bf x}(0))$.
Last, let us assume to pick $t_1 = t_i + i \Delta t$, with $i \in \mathbb{N}_0$ and that a control action is taken at every $t_1$. In view of the above, the optimal control action is $P({\bf u}(0)|{\bf x}(0))\,\forall t_1$, as long as ${\bf x}(t_1) = {\bf x}(k=0) \in \mathcal{X}$, that is true by construction.

\begin{remark}
Depending on applications, the control action/state space dimensions of the  problem can be reduced. For example, in the context of connected cars, as shown in \cite{https://doi.org/10.48550/arxiv.2212.02467} the state space can be conveniently reduced to only the links/states that are reachable by the vehicle within the time horizon. We leave for future research the problem of finding applications-agnostic methods to reduce the dimensionality of the control action/state space suitable for Algorithm \ref{alg:discrete_fpd}.
% The downside of the approach in Remark \ref{prp:time_invariance} is that, if the state space is large, the off-line planning approach is infeasible in any practical situation. In this case, a receding horizon strategy might still be the most effective option, possibly with a set constraint like ${\bf x}(0) \in \mathcal{X}_0$, where $\mathcal{X}_0$ is an arbitrarily large neighborhood of the current state ${\bf x}_0$. This would require re-executing the receding horizon optimization only when ${\bf x}(0)$ gets far enough from ${\bf x}_0$.
\end{remark}

\section{The Proposed Design Pipeline: From Data to Control Inputs}
\label{sec:pipeline}

Before proceeding with the illustration of the pipeline to compute the inputs required by Algorithm \ref{alg:discrete_fpd}, we note here that, as for any other control approach that relies solely on the available data, these need to be sufficiently {informative}. In this paper we do not consider the problem of obtaining sufficiently informative datasets and refer to e.g.\ \cite{8960476,8933093,9062331,COLIN2020109000} for recent results on this topic.

The first step of the pipeline is to gather the data, which are then processed to obtain the probability functions required as input by Algorithm \ref{alg:discrete_fpd}. The last, optional, step of the pipeline consists in formalizing the constraints for Problem \ref{prb:global_opt}.

\subsection{Data gathering}

We refer to the system under control as \emph{target system} and we term as {\em reference system} the one that is used to collect example data. In what follows, we do not require any knowledge on the (possibly) nonlinear and stochastic dynamics that is generating the data. Also, data for the target and reference system might be generated from different unknown dynamics. Our starting point is the collection of the data to compute $\tilde{Q}_X^k$, $\tilde{P}_X^k$ and $\tilde{Q}_U^k$. Namely, we assume the availability of the following data recorded within the observation window $[0,T]$: the triplets $\lbrace {\bf x}^t(k), {\bf u}^t(k), {\bf x}^t(k+1)\rbrace$ and $\lbrace {\bf x}^r(k), {\bf u}^r(k), {\bf x}^r(k+1)\rbrace$, and the pairs $\lbrace {\bf x}^r(k), {\bf u}^r(k) \rbrace$, where ${\bf x}^t(k)$ and ${\bf x}^r(k)$ are the target and reference systems' states and ${\bf u}^t(k)$ and ${\bf u}^r(k)$ are the target and reference systems' inputs at $t(k)$, respectively. Before illustrating how these data are used, we remark here that for physical applications (modeled via continuous dynamics) the data give a view of a discretized version of the processes. Hence, the probability functions that we obtain implicitly depend on the (discretization) step $\Delta t$. From the practical viewpoint, this parameter needs to be in accordance with the duration of the full control cycle.

\begin{remark}
For time-invariant, we drop the superscripts in the above probability functions as these are stationary. Due to the effects of the discretization and quantization, stationary probability functions might, in principle, still be obtained from time-varying dynamical processes.
\end{remark}

\subsection{Quantization}

Once the data are obtained, these need to be quantized to obtain discrete sets of the form \eqref{sets} required by Algorithm \ref{alg:discrete_fpd}. We let  ${\bf \Delta x}$ and ${\bf \Delta u}$ be the uniform quantization steps for state and input. Then, each given observation of the state, say ${\bf x}$, is mapped onto ${\bf x}_j$ defined in \eqref{sets} if 
\begin{equation} \label{discrete_state_interval}
{\bf x}_j - \frac{\bf \Delta x}{2} \leq {\bf x} < {\bf x}_j + \frac{\bf \Delta x}{2}.
\end{equation}
Analogously, a given observed input, say ${\bf u}$, is mapped onto ${\bf u}_h$ defined in \eqref{sets} if
\begin{equation} \label{discrete_input_interval}
{\bf u}_h - \frac{\bf \Delta u}{2} \leq {\bf u} < {\bf u}_h + \frac{\bf \Delta u}{2}.
\end{equation}
At the boundaries, ${\bf x}$ is mapped onto ${\bf x}_0$ when ${\bf x} < {\bf x}_0 + \Delta {\bf x}/2$ and onto ${\bf x}_{m-1}$ when ${\bf x} \geq {\bf x}_{m-1} - \Delta {\bf x}/{2}$. Equivalent considerations hold for the inputs.

\begin{remark}
An interesting open problem, which we leave for future research, is to characterize how control performance are affected by discretization (see also our conclusions).
%While it is reasonable to assume that the reference and target systems share a common state space (i.e.\ the two systems should have similar capabilities in order to deliver similar behaviors), inputs can be quite different, depending on the systems' dynamical characteristics. Thus input data might need to be normalized prior to the application of \eqref{discrete_input_interval}.
\end{remark}

%In \eqref{state_set} and \eqref{input_set}, the state and the input domains have been discretized. While it is reasonable to assume that the reference and target systems will share a common state space (i.e.\ the two systems should have similar capabilities in order to deliver similar behaviors), inputs can be quite different, depending on the systems' dynamical characteristics. Together with the input domain limits, resolution also plays an important role, as it affects the time and space complexity of the discrete FPD optimization. Even in terms of granularity of the available control inputs, the reference and target systems have, in general, different requirements. Depending on the specific characteristics of the systems at hand, normalization techniques can be designed to map the systems' input domains to a common domain.

%Said that the discrete domains should be designed in such a way that no state and no input can be found in the data that goes beyond the domain limits, in some circumstances, this can be a design choice. In these cases, it is worth defining different quantization conditions than \eqref{discrete_state_interval} and \eqref{discrete_input_interval}: a given observation ${\bf x}$ is mapped onto ${\bf x}_0$ when
%\begin{equation}
%{\bf x} < {\bf x}_0 + \frac{\Delta {\bf x}}{2}
%\end{equation}
%and onto ${\bf x}_{m-1}$ when
%\begin{equation}
%{\bf x} \geq {\bf x}_{m-1} - \frac{\Delta {\bf x}}{2}.
%\end{equation}
%Equivalent considerations hold for the inputs.

\subsection{Computation of the probability functions}

As in e.g.\ \cite{Gagliardi}, we obtain the probability functions by computing the empirical distributions from the data. This can be achieved by defining three counting functions: (i) $c_{X}:\mathcal{X} \rightarrow \mathbb{N}$ counts the occurrences of a given state in the collected reference (or target) system data; (ii) ${\bf c}_{X|X,U}:\mathcal{X} \times \mathcal{U} \rightarrow \mathbb{N}^m$, for each visited state and selected input in the same state, counts the occurrences of every state in the following time instant in the collected reference (or target) system data; (iii) ${\bf c}_{U|X}:\mathcal{X} \rightarrow \mathbb{N}^z$, for each visited state, counts the occurrences of every selected input in the collected reference system data. The probabilistic functions for the target thus be computed, for each $i$ and $h$, as:
\begin{equation}\label{sem_computation_translated}
\tilde{Q}_X({\bf x} | {\bf x}_i, {\bf u}_h)  = \frac{{\bm \tau}_{X|X,U}({\bf x}_i, {\bf u}_h)}{{\bm \tau}_{U|X}^h({\bf x}_i)}, \ \
\tilde{Q}_U({\bf u} | {\bf x}_i)  = \frac{{\bm \tau}_{U|X}({\bf x}_i)}{\tau_X({\bf x}_i)}.
\end{equation}
Since the counting functions are vector functions, the superscript indicates the element in the vector.
In the above expressions, we have $\tau_{X}({\bf x}) := o_s + c_{X}({\bf x}), {\bm \tau}_{X|X,U}({\bf x}, {\bf u}) := o_n + {\bf c}_{X|X,U}({\bf x}, {\bf u})$ and ${\bm \tau}_{U|X}({\bf x}) := o_i + {\bf c}_{U|X}({\bf x})$, with $o_s$, $o_n$ and $o_i$ being small constants (offsets) added in order to avoid divisions by $0$ in \eqref{sem_computation_translated} which would happen when a given event is not seen in the data, as well as the computation of $\log{(0)}$ for, e.g., the coefficients in Problem \ref{prb:local_opt}. Clearly, the same steps can be followed to obtain $\tilde{P}_X({\bf x} | {\bf x}_i, {\bf u}_h)$ and these are omitted here for brevity.
\begin{remark}
A possible choice for the  offsets $o_s$ is to set $o_s \le 1/m$. Once $o_s$ is fixed, all the other offsets must be designed so as to guarantee that probabilities sum to one. Then one needs to set $o_i = o_s/z$ and $o_n = o_i/m$.
\end{remark}
\begin{remark}
When an input is never selected in a given state or a state is never visited, the data do not contain any information about $\tilde{Q}_X({\bf x} | {\bf x}_i, {\bf u}_h)$ and $\tilde{Q}_U({\bf u} | {\bf x}_i)$ respectively. In these cases the best we can do is assuming that any state can be reached at the next time step with equal probability by applying ${\bf u}_h$ in ${\bf x}_i$ and that any input is selected with equal probability when in ${\bf x}_i$, respectively. Equations in \eqref{sem_computation_translated} automatically implement this condition.
\end{remark}

\subsection{Respecting the range of feasible inputs}
 
Problem \eqref{prb:local_opt} accounts for soft and hard constraints imposed through probabilities. Here we discuss about the possibility of manipulating the input domains so as to respect the range of feasible inputs, regardless of the optimization performed by Algorithm \ref{alg:discrete_fpd}. Let us assume that the reference and target systems are commanded with inputs ${\bf v}_{r,t}$, and that the (symmetrical) ranges of feasible inputs are $-\overline{\bf v}_{r,t} \leq {\bf v}_{r,t} \leq \overline{\bf v}_{r,t}$, where $\pm \overline{\bf v}_{r,t}$ represent the boundaries, extracted, for instance, from the plants' datasheet. The inputs of the probabilistic model are the normalized ones, i.e.\
\begin{equation} \label{normalization}
{\bf u}_r = \frac{{\bf v}_r}{\overline{\bf v}_r}, {\bf u}_t = \frac{{\bf v}_t}{\overline{\bf v}_t}
\end{equation}
which implies $-{\bf 1} \leq {\bf u}_{\lbrace r,t \rbrace} \leq {\bf 1}$, with ${\bf 1}$ representing the vector of ones of appropriate size.

Once the control policy $\tilde{P}_k^U$ is available, the probabilistic controller reconstructs the target system input by inverting \eqref{normalization}, hence setting ${\bf u}_0 = -{\bf 1}, {\bf u}_{z-1} = {\bf 1}$ in \eqref{sets} ensures that the ranges of feasible inputs are respected without defining explicit constraints in Problem \ref{prb:global_opt}.

\subsection{Adding constraints in the optimization problem}

This last step in the pipeline is optional and is only requested if the problem has actuation constraints, beyond respecting the range of feasible inputs. We recall that the constraint functions that the DFPD takes as input are generic and one only needs to verify that the set is convex (recall that convexity of the set guarantees the existence of an optimal solution to Problem \ref{prb:global_opt}). Constraints of practical interest, which can be captured with a proper choice of constraint functions, include moment constraints \cite{1246014,1618839,8359301} and bound constraints \cite{8651519,9290355}. For example, the inequality constraint $\E_{\tilde{P}^k_{U}}[U^i] - m_i \le 0$ expresses the fact that the $i$-th moment of the control input is not larger than $m_i$. Bound constraints instead formalize the fact that $\mathbb{P}(\bv{U}_k\in\bar{\mathcal{U}}_k)\ge 1-\varepsilon$,  where $\bar{\mathcal{U}}_k \subset \mathcal{U}$ is and $\varepsilon \ge 0$. That is, the constraint captures the fact that the control variable belongs to some (e.g.\ desired) $\bar{\mathcal{U}}_k$ with some desired probability. This constraint can be included in Algorithm \ref{alg:discrete_fpd} as $\E_{\tilde{f}_\bv{U}^k}\left[\mathds{1}_{\bar{\mathcal{U}}_k}(\mathbf{U}_k)\right]$.

\section{Benchmarking the DFPD  on an Inverted Pendulum}\label{sec:numerical_results}

We now numerically investigate the effectivenes of the DFPD  by using an inverted pendulum with actuation constraints as test-bed. We first describe the environment, then we describe how the inputs to Algorithm \ref{alg:discrete_fpd} are obtained. We finally discuss the numerical results. The fully documented code to replicate the results can be found at \url{https://github.com/unisa-acg/discrete-fpd}.

\subsection{The environment}

We consider the task of stabilizing a pendulum on its unstable equilibrium. The pendulum we want to stabilize (i.e.\ the target system) has parameters that are different from the one used to collect the data (i.e.\ the reference system). Specifically, the parameters of the reference system were as follows: rod length, $l_r=0.2$ m, mass, $m_r=0.5$ kg, friction $b_r = 8 \cdot 10^{-5}$ Nm/(deg/s). Instead, for the target pendulum we have $l_t=0.4$ m, $m_t=1$ kg, $b_t = 1 \cdot 10^{-5}$ Nm/(deg/s). As usual, the input to the pendulum is the torque applied to its hinged end, i.e.\ ${\bf u}=u$, while the state is ${\bf x} = \left[ x_1, x_2 \right]^T$, with  $x_1$ being the angular position and $x_2 = \dot{x}_1$ the angular velocity. Finally, in one of our simulations, we impose the constraint that the target system cannot use a torque that is larger than $50\%$ of the torque capacity, say $\tau_{t,max} = 11.5$ Nm, i.e.\ $-0.5 \leq {\bf u} \leq 0.5$.

\subsection{Obtaining the inputs to Algorithm \ref{alg:discrete_fpd}.}

Data were obtained by simulating the nonlinear dynamics of the target and the reference system.\footnote{Mathematical models are only used to generate the data and are not leveraged to compute the control action} The dynamics of the two systems were stochastic: zero mean Gaussian noise processes with variances $\sigma_r^2=20$ and $\sigma_t^2=10$ were additively added to the acceleration of the reference and target dynamics, respectively. Data were generated by making the reference system follow a path that takes it from the stable equilibrium state ${\bf x} = \left[ -\pi/2, 0 \right]^T$ to the unstable equilibrium state ${\bf x} = \left[ \pi/2, 0 \right]^T$. Examples were generated via a model-based PID controller, simulating a number of control policies that differed in the time law used by the pendulum to reach the target position. The time parametrization was obtained through the phase plane technique, see e.g. \cite{Slotine}. We used this technique as it allowed us to formulate the trajectory generation problem for the pendulum with a reduced set of parameters. In fact all the trajectories belong to one of two classes having either one or three characteristic switching points (simply switching points in what follows -- see \cite{Slotine} for the precise definition) in the phase plane, that could be suitably assigned through randomized uniform sampling in a given range. The switching points were then connected through linear interpolation. Figure \ref{fig:reference_data_single} and Figure \ref{fig:reference_data} provide two different views of reference system's position, velocity and torque signals over 100 simulations, where $30\%$ of trajectories belong to the 1-switching-point class, while the remaining ones belong to the 3-switching-point class. In order to collect data on the state evolution model of the target pendulum, we provided it with open loop torques that excited the system at different configurations. All simulations for the reference and target systems were executed with $\Delta t = 0.01$ s.

\begin{figure}
    \centering   
    \includegraphics[width=0.32\textwidth]{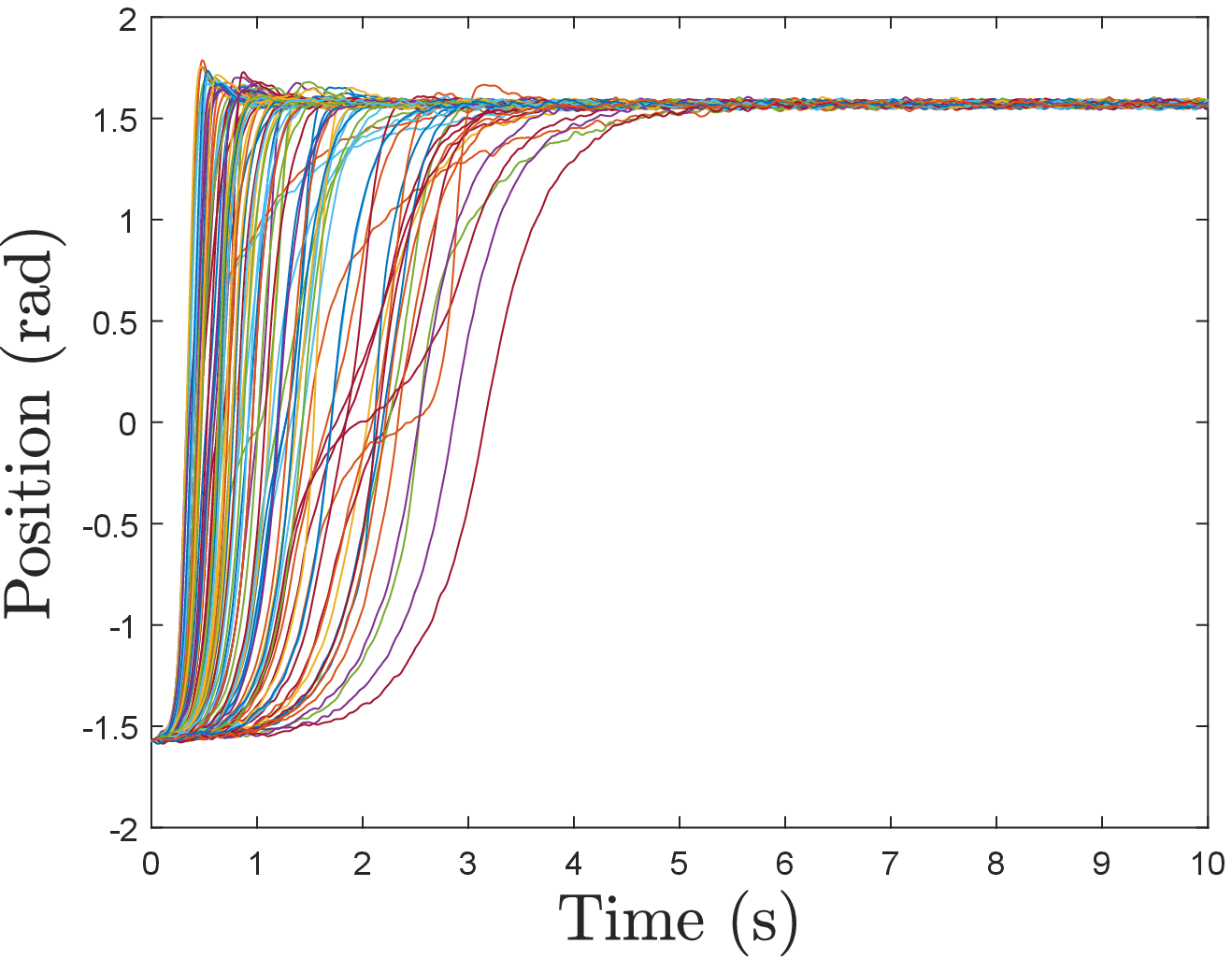}
    \includegraphics[width=0.32\textwidth]{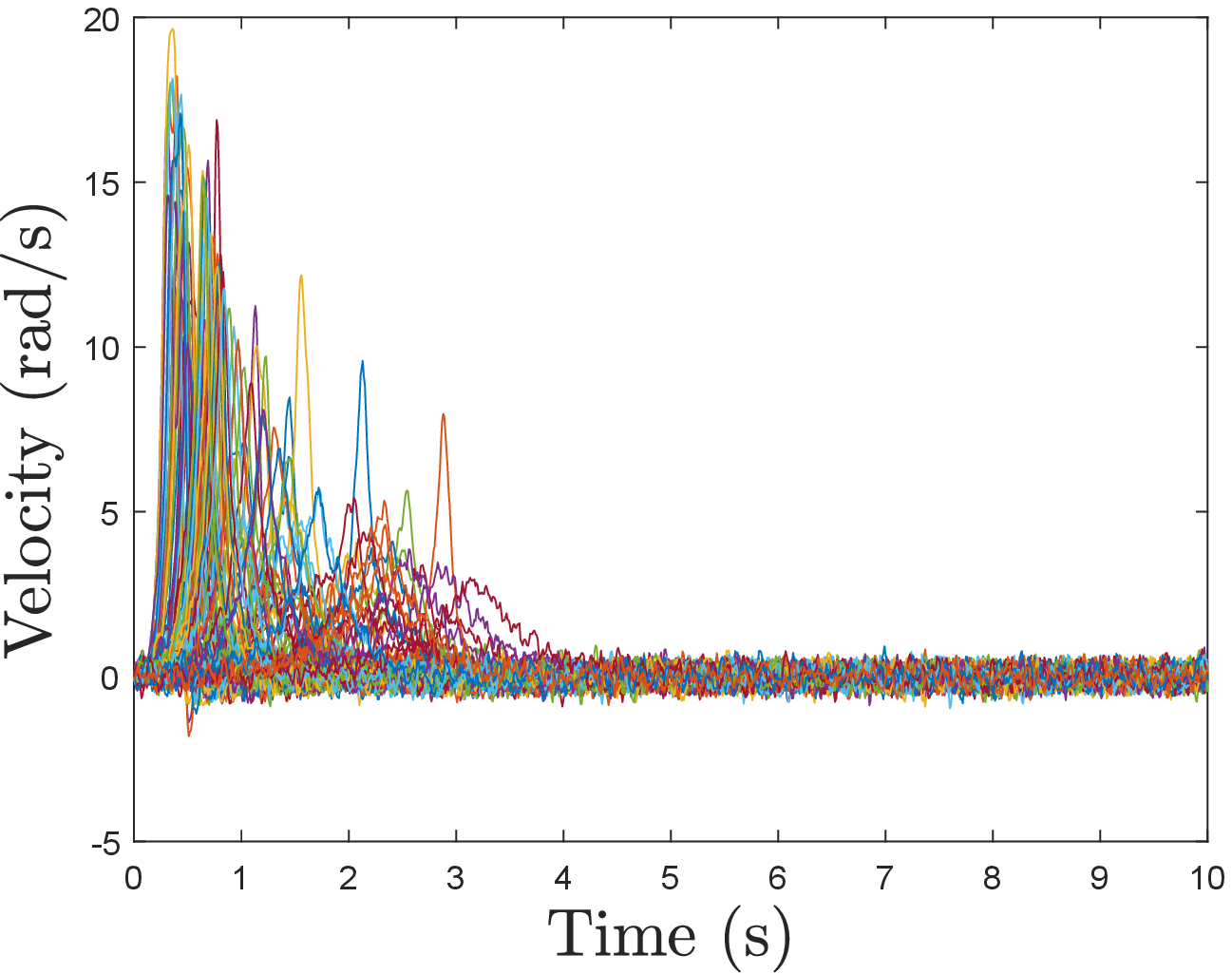}
    \includegraphics[width=0.32\textwidth]{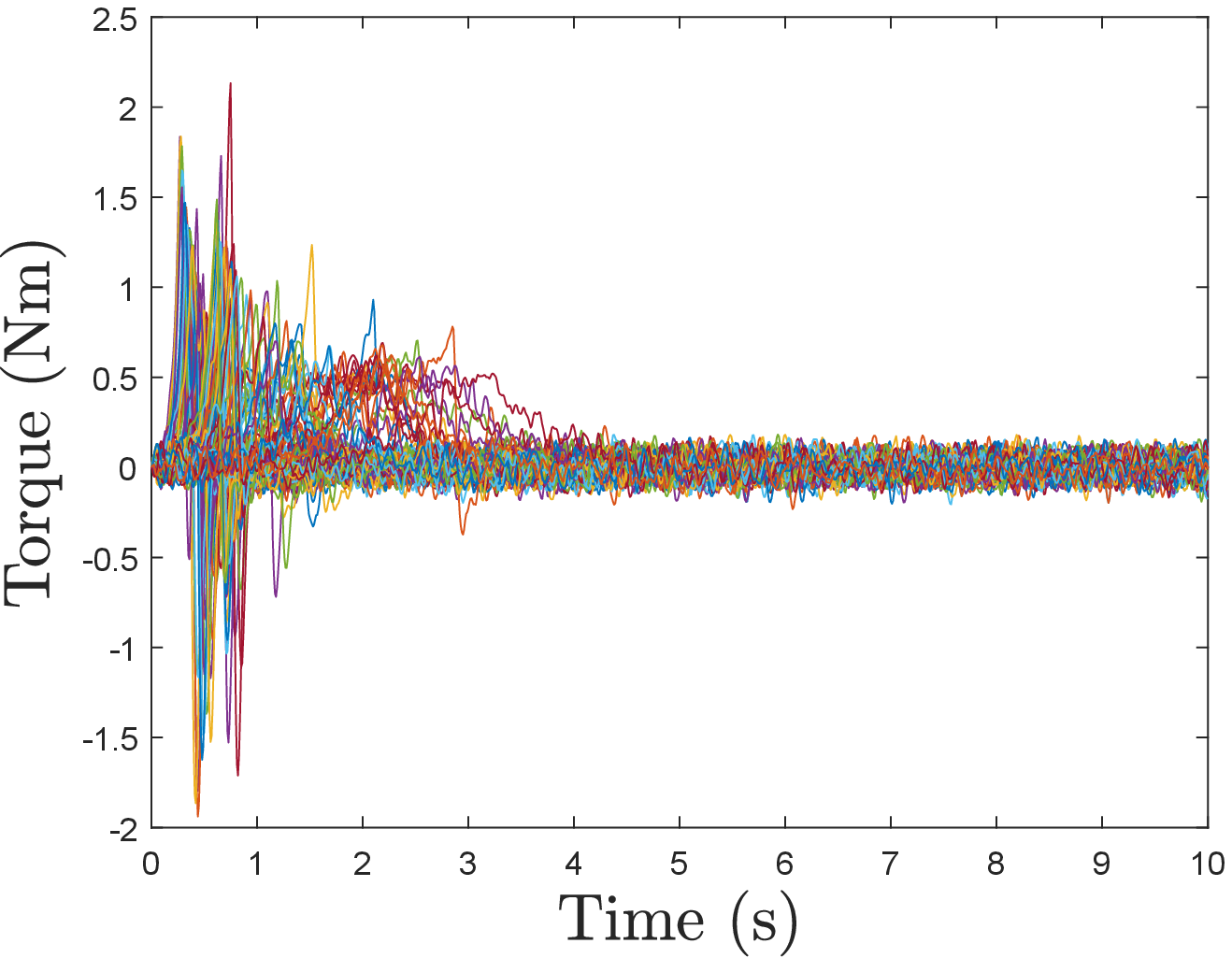}
    \caption{Reference system data (single trajectories) over 100 simulations.}   
	\label{fig:reference_data_single}
\end{figure}

\begin{figure}
    \centering   
    \includegraphics[width=0.33\textwidth]{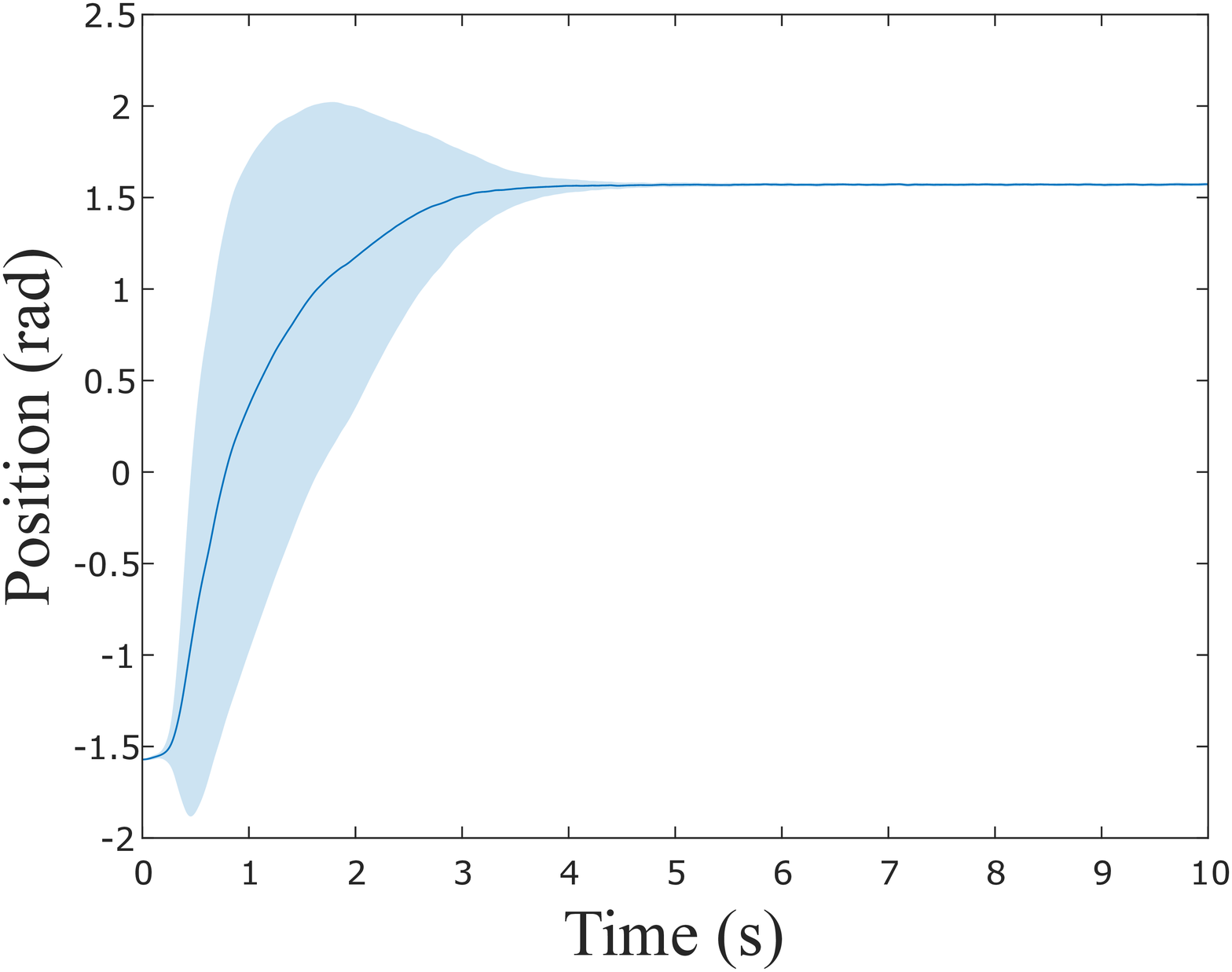}
    \includegraphics[width=0.325\textwidth]{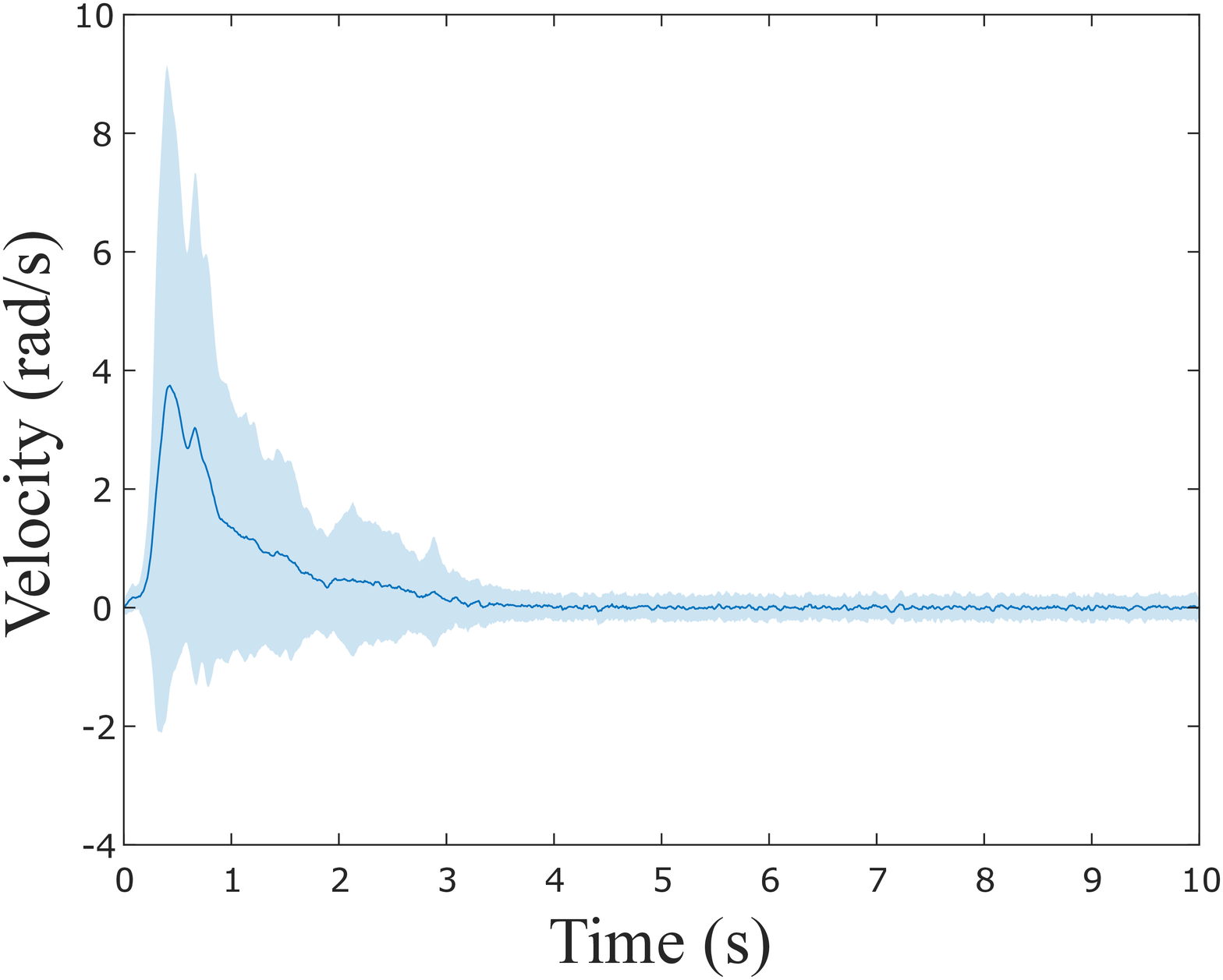}
    \includegraphics[width=0.33\textwidth]{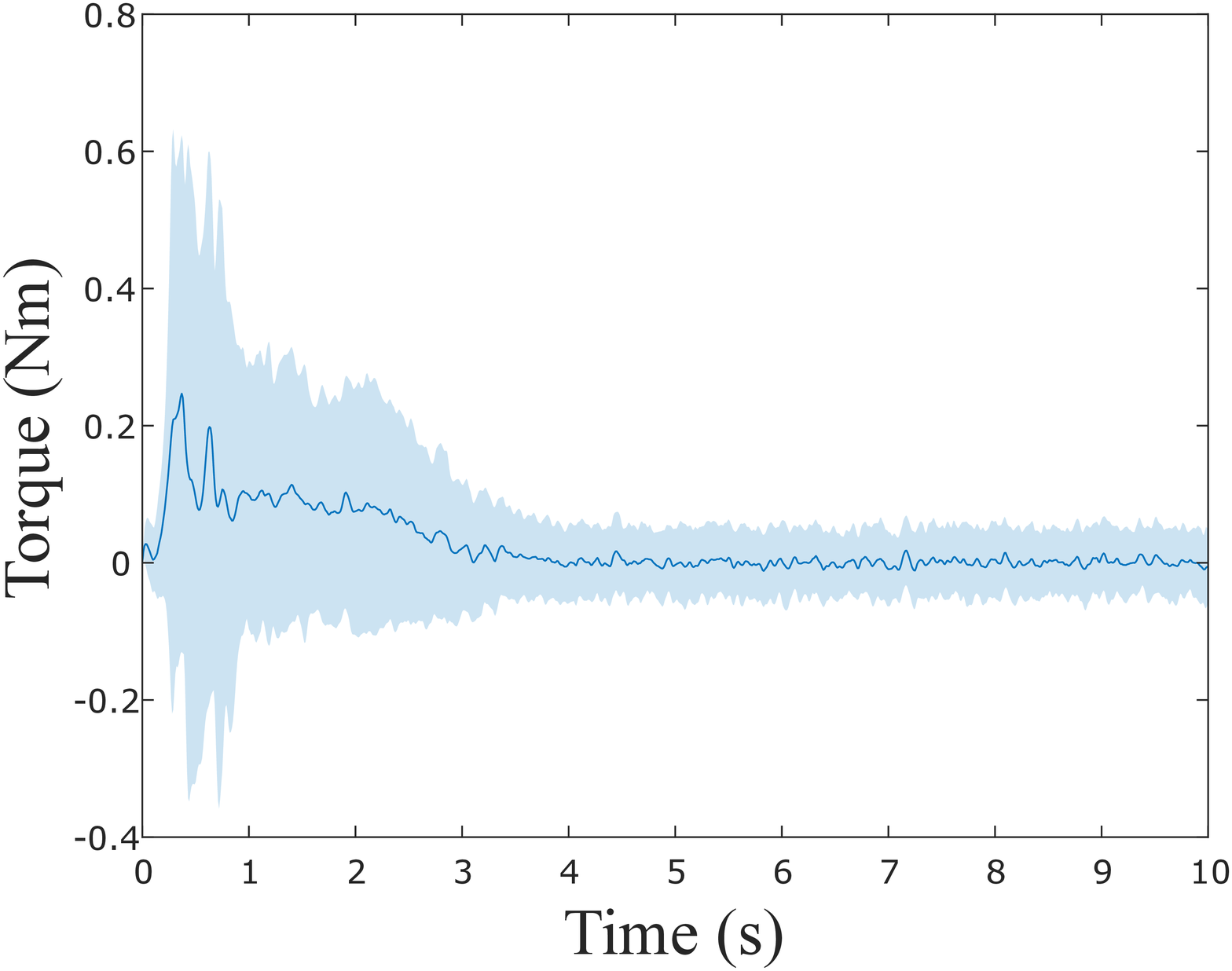}
    \caption{Reference system data over 100 simulations (same as Figure \ref{fig:reference_data_single}); bold lines denote means, shaded areas the confidence intervals corresponding to the standard deviation.}
	\label{fig:reference_data}
\end{figure}

The generation of the probabilistic model followed the arguments of Section \ref{sec:pipeline}. First, the state and input domains were discretized, then the probabilistic models were computed as in \eqref{sem_computation_translated}, assuming time invariance. Discrete state domain boundaries were determined on the basis of the maximum and minimum values observed in the reference and target data, thus ${\bf x}_0 = \left[-1.5872, -1.8107\right]^T$ and ${\bf x}_{m-1} = \left[1.9583, 19.6537\right]^T$. The state discretization step $\Delta {\bf x}$ was selected from a trade-off between accuracy and computation time. In our numerical experiments we set $\Delta {\bf x} = \left[0.1223, 0.7402\right]^T$, as this empirically yielded satisfactory results. For the inputs, we observed that the target system is bigger and heavier than the reference one, from which one should expect higher torques. In our implementation, we considered symmetric torque limits and inputs of the reference system were normalized as $u = \tau_r/\tau_{r,max}$, while those of the target system as $u = \tau_t/\tau_{t,max}$, so that $u$ takes the common semantics of an effort capacity. In these expressions $\tau_r$ and $\tau_t$ are the observed torques of the reference and target system respectively, $\tau_{r,max}$ and $\tau_{t,max}$ are the maximum torques of the reference and target system respectively, as observed from data. Finally we set $\Delta u = 0.0513$ and, in \eqref{sets}, $u_0 = -1$ and $u_{z-1} = 1$, so as to respect the ranges of feasible inputs by construction. Given this set-up, we computed the probability functions $\tilde{Q}_X$, $\tilde{Q}_U$ and $\tilde{P}_X$ needed by Algorithm \ref{alg:discrete_fpd}. For the last experiment, as last input to Algorithm \ref{alg:discrete_fpd}, we defined a constraint on the torque of the pendulum (see the first paragraph of this section). The constraint was captured via a bound constraint.

\subsection{Results}

We obtained the optimal policy to control the target system from data collected from the reference system via Algorithm \ref{alg:discrete_fpd}. In particular, by leveraging the stationarity of the system and of the constraints, we used the observation given in Remark \ref{prp:time_invariance} to retrieve the optimal policy. Namely, once the problem was solved off-line, we sampled, at each $k$, the control input from the conditional probability function $\tilde{P}^0_U$. This policy was the one used at run-time to generete the actual torque inputs to the pendulum. Specifically, at each control cycle: (i) the current state (pendulum position and velocity) is read; (ii) quantization is retrieved; (iii) the right histogram for $\tilde{P}^0_U({\bf u}|{\bf x}_j)$ is queried from the memory (note indeed that $\tilde{P}^0_U({\bf u}|{\bf x}_j)$ is conditioned and hence the histogram depends on the current state of the pendulum). From the histogram, the control input with the highest-probability value was selected and the torque was obtained as $\tau_t = u \cdot \tau_{t,max}$. The discrete FPD optimization was executed with $n=10$. The pendulum position, velocity and torque obtained with the FPD control policy are reported in blue in Figure \ref{fig:comparison} (colors online): the data driven control policy correctly stabilizes the pendulum around the state ${\bf x} = \left[ \pi/2, 0 \right]^T$, imitating the behavior of the reference system (instead shown in Figure \ref{fig:reference_data}). In Figure \ref{fig:comparison} we also report, in red, the result of using the same  policy on the reference system. The figure shows that, when the policy from the reference system is exported to control the target system, this is not able to stabilize the pendulum on the desired equilibrium. The reason for this is in the fact that the parameters of the two systems were different. In these simulations there were no constraints and our last validation step consists in verifying the ability of the DFPD to stabilize the pendulum even in the presence of the constraints mentioned in the first paragraph of this section. In Figure \ref{fig:data_driven_results_con} the behavior of the closed loop system is shown when these constraints were added. The figure clearly shows that DFPD was able to stabilize the pendulum in the presence of these constraints.

\begin{figure}
    \centering   
    \includegraphics[width=0.33\textwidth]{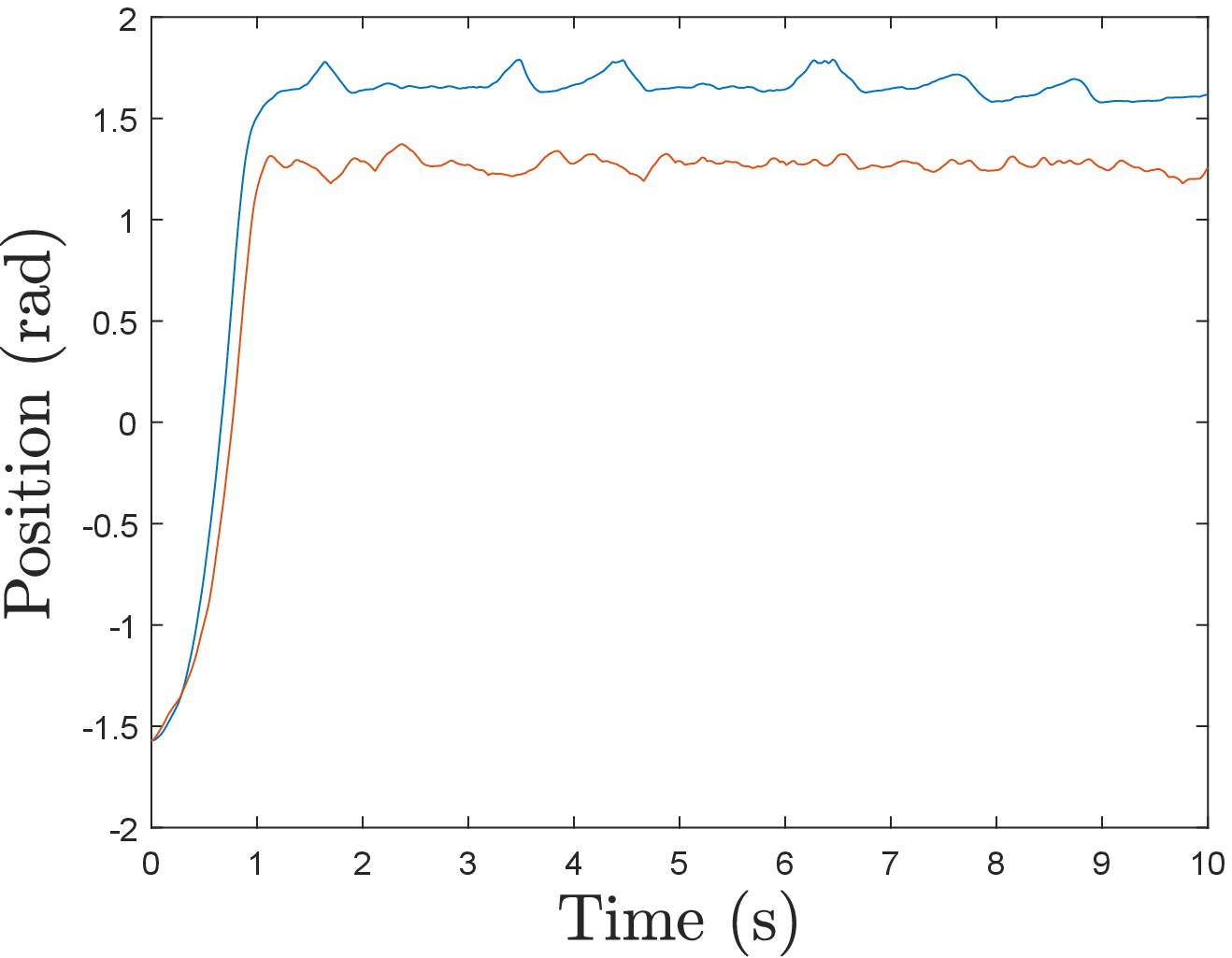}
    \includegraphics[width=0.32\textwidth]{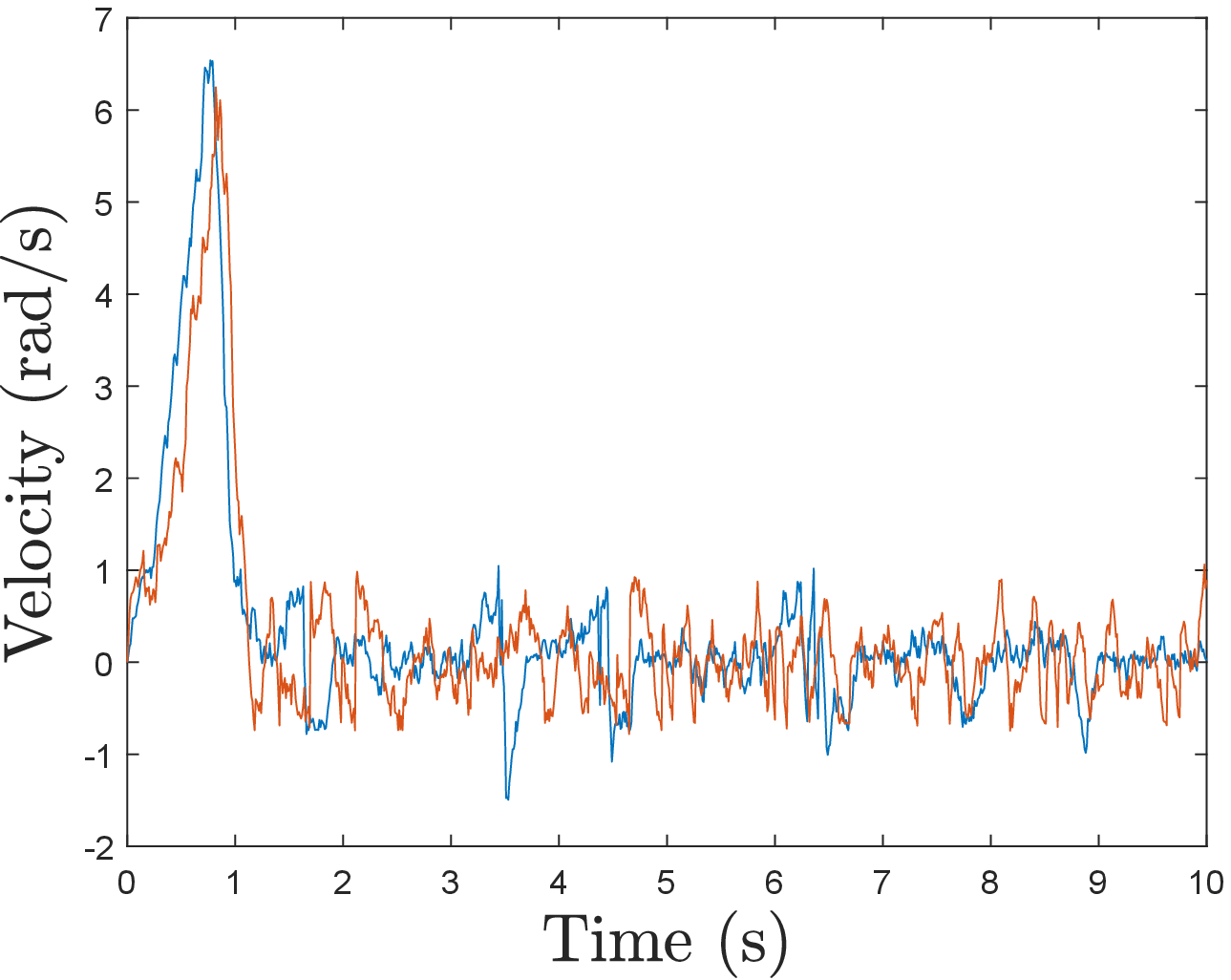}
    \includegraphics[width=0.32\textwidth]{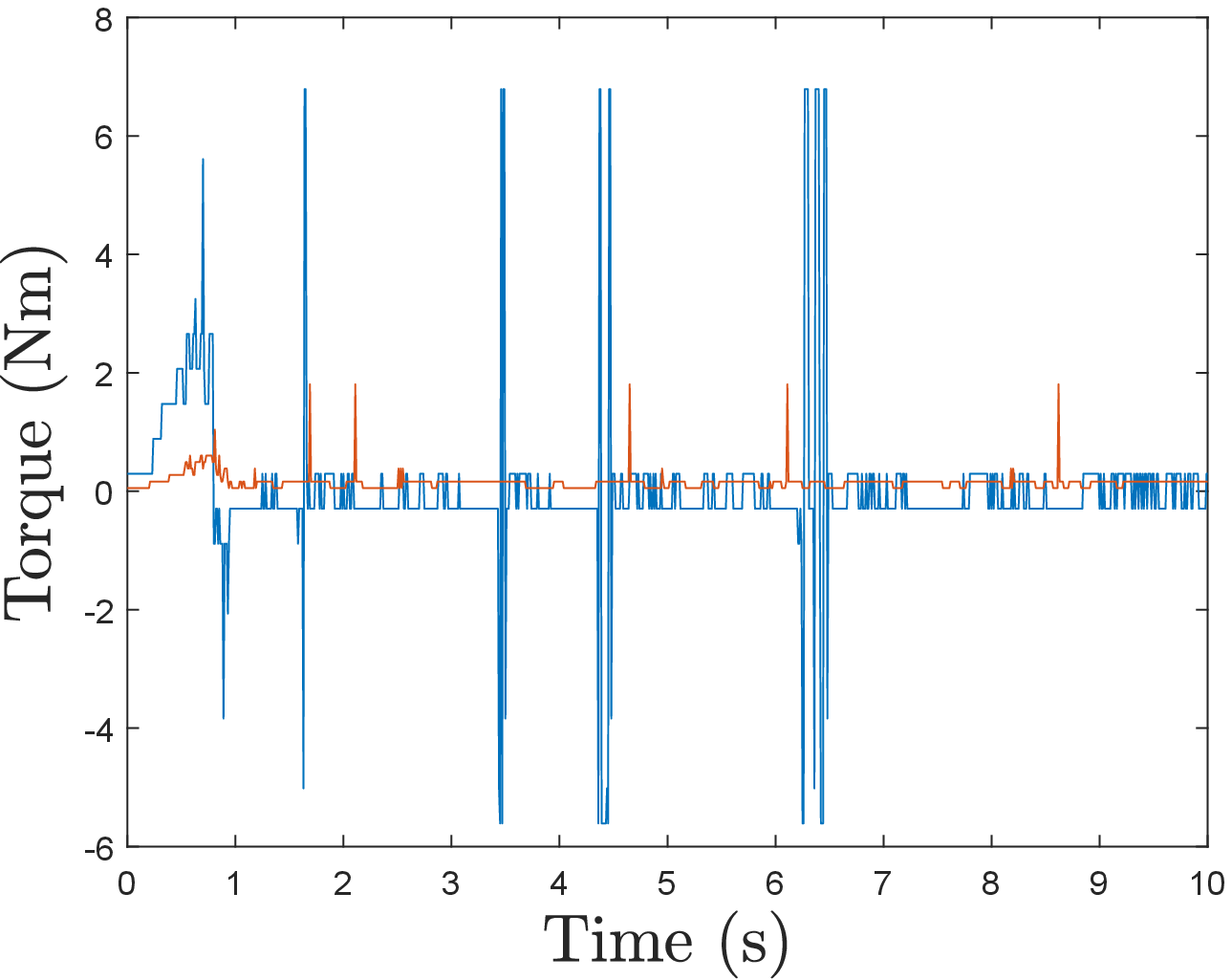}
    \caption{Time evolution of position, velocity and torque when the data-driven control policy is used to control the target (blue) and reference (red) system. Colors online.}   
	\label{fig:comparison}
\end{figure}

\begin{figure}[!htb]
    \centering   
    \includegraphics[width=0.33\textwidth]{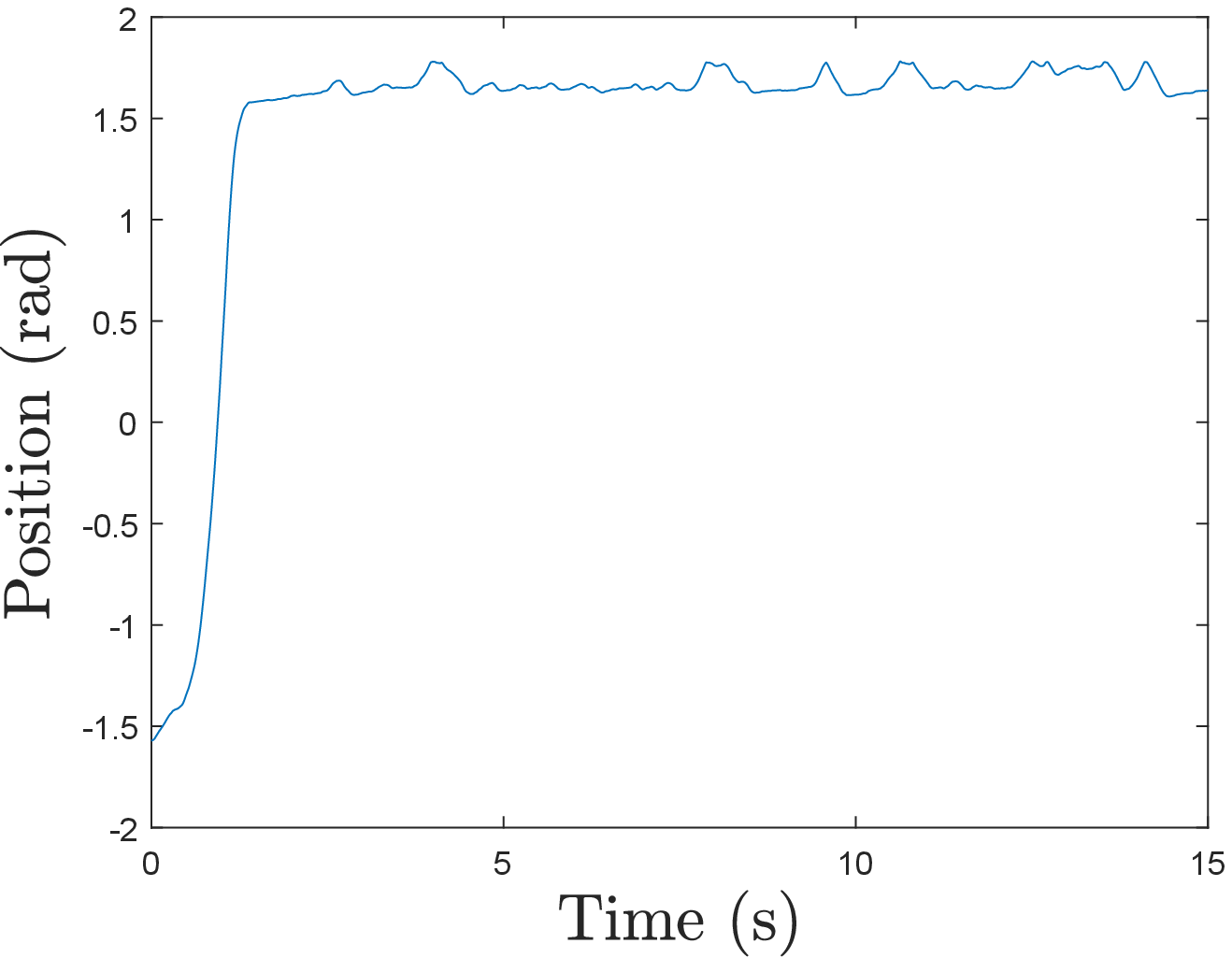}
    \includegraphics[width=0.32\textwidth]{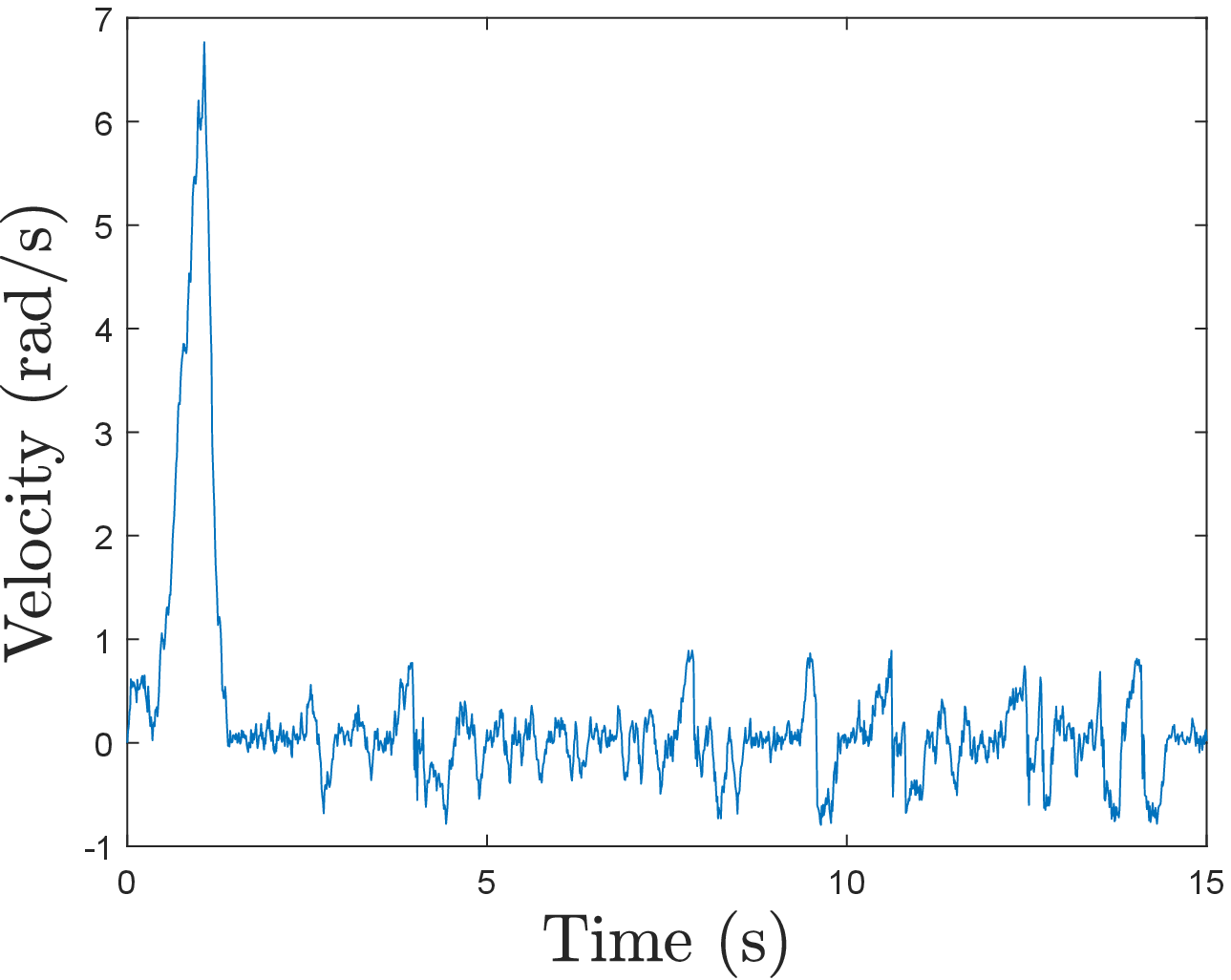}
    \includegraphics[width=0.32\textwidth]{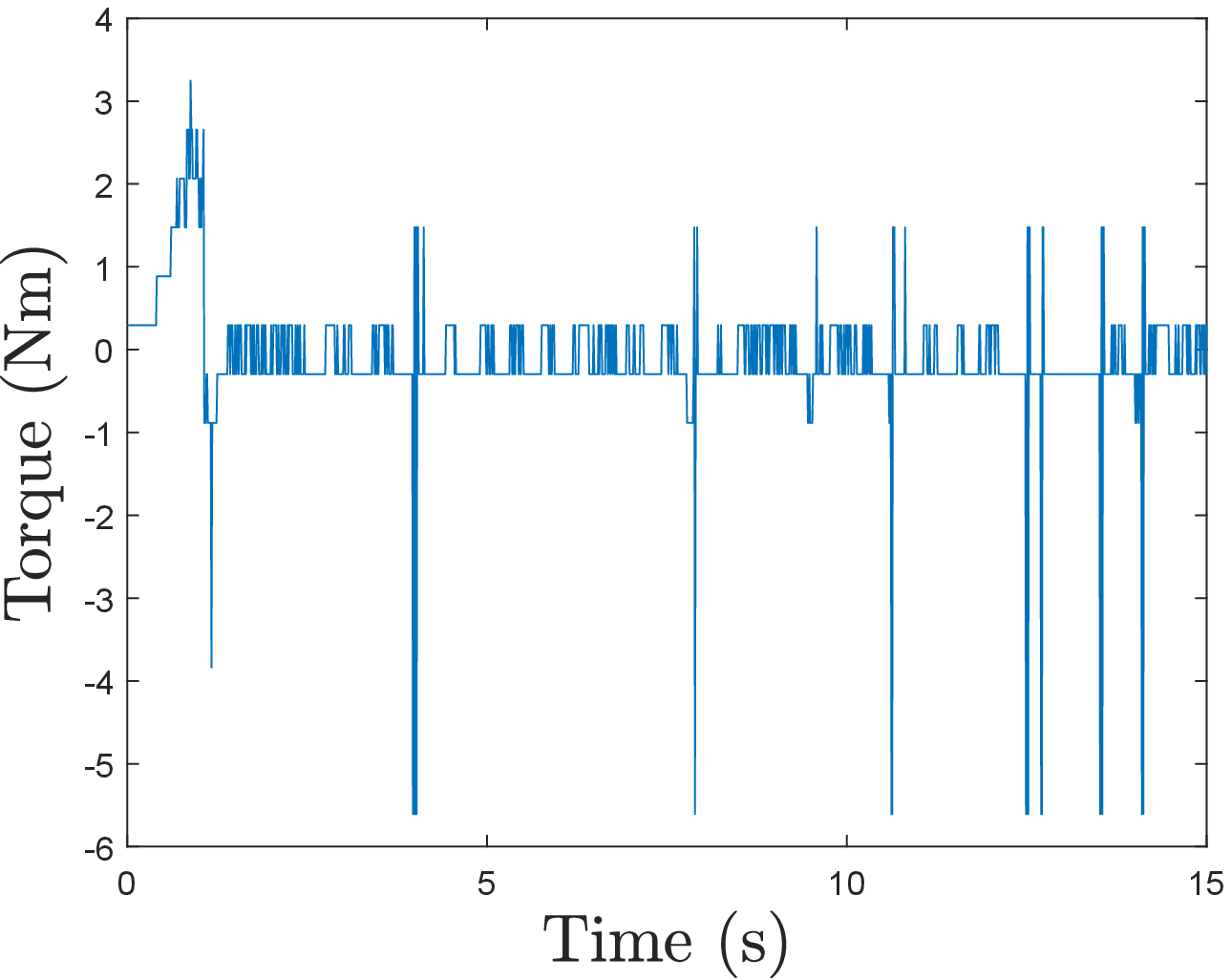}
    \caption{Target pendulum trajectory through data-driven control with restricted torque constraints}   
	\label{fig:data_driven_results_con}
\end{figure}

\section{Conclusions}\label{sec:conclusions}

Motivated by the problem of designing control policies from example data for constrained systems having a possibly stochastic and nonlinear dynamics, we introduced the principled design of a pipeline that enables control synthesis in these situations. The pipeline expounds an algorithm from \cite{Gagliardi} for the synthesis of control policies from data collected from a system that is different from the one under control, without requiring that the constraints are already  fulfilled in the possibly noisy example data. We benchmarked numerically the DFPD on an example that involves controlling an inverted pendulum. The pendulum was affected by actuation constraints and the demonstration data were collected from a physically different pendulum that does not satisfy these constraints. Our simulations highlighted that DFPD is able to stabilize the pendulum directly from the data while satisfying the constraints embedded in the problem formulation. The fully documented code implementing our design is also made openly available. 

\noindent We plan to build on the pipeline presented here to develop a fully end-to-end pipeline for control synthesis from demonstrations. In doing so, we also plan to characterize how discretization affects control performance and to devise control action/state dimensionality reduction techniques that can be useful to scale the approach we presented. Finally, we aim to deliver a large scale demonstrator of the end-to-end pipeline.

\appendix
\section{Sketch of the derivation of Problem \ref{prb:local_opt}}
\label{apx:derivation}

The derivations are adapted from \cite{Gagliardi} and we refer to such a paper for the rigorous arguments. Here we only focus on the cost in Problem \ref{prb:local_opt} as the constraints are obtained immediately from Problem \ref{prb:global_opt}. Following the chain rule for the KL-divergence, the cost of Problem \ref{prb:global_opt} can be rewritten 
\begin{equation} \label{opt_recursive}
\begin{split}
\min_{\tilde{P}_U^{0}, \ldots, \tilde{P}_U^{n-2}} \Bigg\{ & \mathcal{D}_{KL}(P^{n-1} || Q^{n-1}) + \\
& + \min_{\tilde{P}^{n-1}_U} \mathbb{E}_{P^{n-1}} \left[ \mathcal{D}_{KL}(\tilde{P}^{n} || \tilde{Q}^{n}) \right] \Bigg\}.
\end{split}
\end{equation}
The fact that $\mathcal{D}_{KL}(\tilde{P}^{n} || \tilde{Q}^{n})$ is only a function of ${\bf x}(n-1)$ implies that the expectation in the inner minimization can be taken over ${P \left( {\bf x}(n-1) \right) }$. Linearity of the expectation together with the fact that the decision variable of the inner optimization does not depend on $P \left( {\bf x}(n-1) \right)$, imply that \eqref{opt_recursive} can be recast as
\begin{equation} \label{opt_recursive_with_d}
\min_{\tilde{P}_U^{0}, \ldots, \tilde{P}_U^{n-2}} \left\{ \mathcal{D}_{KL}(P^{n-1} || Q^{n-1}) +   \mathbb{E}_{P \left( {\bf x}(n-1) \right)} \left[ d \big( {\bf x}(n-1) \big) \right] \right\}
\end{equation}
where $d \big( {\bf x}(n-1) \big) := \min_{\tilde{P}^{n-1}_U} \mathcal{D}_{KL}(\tilde{P}^{n} || \tilde{Q}^{n})$. The chain rule applied on $\mathcal{D}_{KL}(P^{n-1} || Q^{n-1})$ allows us to write \eqref{opt_recursive_with_d} as
\begin{equation}
\begin{split}
\min_{\tilde{P}_U^{0}, \ldots, \tilde{P}_U^{n-2}} \bigg\{ & \mathcal{D}_{KL}(P^{n-2} || Q^{n-2}) + \\
& + \mathbb{E}_{P^{n-2}} \left[ \mathcal{D}_{KL}(\tilde{P}^{n-1} || \tilde{Q}^{n-1}) \right] + \\
& + \mathbb{E}_{P\left({\bf x}(n-1)\right)} \left[ d \big( {\bf x}(n-1) \big) \right] \bigg\},
\end{split}
\end{equation}
that, following the same arguments outlined above, can be written as
\begin{equation}
\begin{split}
\min_{\tilde{P}_U^{0}, \ldots, \tilde{P}_U^{n-3}} \Bigg\{ & \mathcal{D}_{KL}(P^{n-2} || Q^{n-2}) + \\
& + \mathbb{E}_{P \left( {\bf x}(n-2) \right)} \bigg[ \min_{\tilde{P}_U^{n-2}} \Big\{ \mathcal{D}_{KL}(\tilde{P}^{n-1} || \tilde{Q}^{n-1}) + \\
& + \mathbb{E}_{P\left({\bf x}(n-1)\right)} \left[ d \big( {\bf x}(n-1) \big) \right] \Big\} \bigg] \Bigg\}.
\end{split}
\end{equation}
It can then be shown that, at each $k$, the inner optimization can be written as:
\begin{equation}
\begin{split}
d \left( {\bf x}(k) \right) := \min_{\tilde{P}_U^{k}} \Bigg\{ & \mathcal{D}_{KL}(\tilde{P}_U^{k} || \tilde{Q}_U^{k}) + \\
& + \sum_{{\bf u}(k)} \tilde{P}_U^{k} \mathcal{D}_{KL}(\tilde{P}_X^{k+1} || \tilde{Q}_X^{k+1} ) + \\
& + \mathbb{E}_{P\left({\bf x}(k+1)\right)} \left[ d \big( {\bf x}(k+1) \big) \right] \Bigg\}.
\end{split}
\end{equation}
We now introduce node indices, yielding the equivalent formulation
\begin{equation}
\begin{split}
d \left( {\bf x}(k) \right) = \min_{\tilde{P}_U^{k}} & \Bigg\{ \sum_{h=0}^{z-1} P\big({\bf u}_h(k)| {\bf x}_i(k) \big) \ln \frac{P\big({\bf u}_h(k)| {\bf x}_i(k) \big)}{Q\big({\bf u}_h(k)| {\bf x}_i(k) \big)} + \\
& + \sum_{h=0}^{z-1} P\big({\bf u}_h(k)| {\bf x}_i(k) \big) \mathcal{D}_{KL}(\tilde{P}_X^{k+1} || \tilde{Q}_X^{k+1} ) + \\
& + \sum_{h=0}^{z-1} P\big({\bf u}_h(k)| {\bf x}_i(k) \big) \\ 
& \sum_{j=0}^{m-1} P\big( {\bf x}_j | {\bf x}_i(k), {\bf u}_h(k) \big) d \big( {\bf x}_j(k+1) \big) \Bigg\},
\end{split}
\end{equation}
from which the cost function of Problem \ref{prb:local_opt} can be derived. It is easy to recognize that the optimization for the last time instant can be easily obtained by setting $d = 0$, from which $r_{hi} = 0$ follows.

\bibliographystyle{IEEEtran}
\bibliography{arxiv-dfpd}

\end{document}